%% file: main.tex
\documentclass[11pt]{article}
\usepackage{latexsym}
\usepackage{amsmath,amssymb,amsthm,mathtools}
\usepackage{amsfonts}
\usepackage{authblk}
\usepackage{cite}
\usepackage{hyperref}
\usepackage[linesnumbered,lined,ruled,noend,vlined]{algorithm2e}
\usepackage{apxproof}
\usepackage{cleveref}
\usepackage{tabularx,environ}
\usepackage{comment}
\usepackage{url}
\usepackage[bold,full]{complexity}
\usepackage{lscape}

\input{macro}

\title{Efficient Enumerations for Minimal Multicuts and Multiway Cuts}  

\author[1]{Kazuhiro Kurita}
\author[2]{Yasuaki Kobayashi}

\affil[1]{National Institute of Informatics, Tokyo, Japan, \texttt{kurita@nii.ac.jp}}
\affil[2]{Kyoto University, Kyoto, Japan, \texttt{kobayashi@iip.ist.i.kyoto-u.ac.jp}}

\date{\empty}

\begin{document}

\maketitle

\begin{abstract}
    Let $G = (V, E)$ be an undirected graph and let $B \subseteq V \times V$ be a set of terminal pairs.
    A node/edge multicut is a subset of vertices/edges of $G$ whose removal destroys all the paths between every terminal pair in $B$.
    The problem of computing a {\em minimum} node/edge multicut is NP-hard and extensively studied from several viewpoints.
    In this paper, we study the problem of enumerating all {\em minimal} node multicuts.
    We give an incremental polynomial delay enumeration algorithm for minimal node multicuts, which extends an enumeration algorithm due to Khachiyan et al. (Algorithmica, 2008) for minimal edge multicuts.
    
    Important special cases of node/edge multicuts are node/edge {\em multiway cuts}, where the set of terminal pairs contains every pair of vertices in some subset $T \subseteq V$, that is, $B = T \times T$.
    We improve the running time bound for this special case: We devise a polynomial delay and exponential space enumeration algorithm for minimal node multiway cuts and a polynomial delay and space enumeration algorithm for minimal edge multiway cuts.
\end{abstract}

\section{Introduction}
    Let $G = (V, E)$ be an undirected graph and let $B$ be a set of pairs of vertices of $V$.
    We call a pair in $B$ a {\em terminal pair} and the set of vertices in $B$ is denoted by $T(B)$.
    A {\em node multicut} of $(G, B)$ is a set of vertices $M \subseteq V \setminus T(B)$ such that there is no path between any terminal pair of $B$ in the graph obtained by removing the vertices in $M$.
    A {\em edge multicut} of $(G, B)$ is defined as well: the set of edges whose removal destroys all the paths between every terminal pair.
    The minimum node/edge multicut problem is of finding a smallest cardinality node/edge multicut of $(G, B)$.
    When $B = T \times T$ for some $T \subseteq V$, the problems are particularly called the minimum node/edge multiway cut problems, and a multicut of $(G, B)$ is called a {\em multiway cut} of $(G, T)$.
    
    These problems are natural extensions of the classical minimum $s$-$t$ separator/cut problems, which can be solved in polynomial time using the augmenting path algorithm.
    Unfortunately, these problems are NP-hard~\cite{Dahlhaus::1994} even for planar graphs and for general graphs with fixed $|T| \ge 3$.
     Due to numerous applications (e.g. \cite{Fireman::2007,Kappes::2011,Stone::1977}), a lot of efforts have been devoted to solving these problems from several perspectives such as approximation algorithms~\cite{Garg::1996,Garg::1997,Arora::1999,Calinescu::2000,Karger::2004}, parameterized algorithms~\cite{Marx::2014,Cygan::2013,Guillemot::2011,Marx::2006,Xiao::2010}, and restricting input~\cite{Guo::2008,Bateni::2012,Chen::2004,Dahlhaus::1994,Klein::2012,Marx::2012}.
    
    In this paper, we tackle these problems from yet another viewpoint, in which our focus in this paper is {\em enumeration}.
    Since the problems of finding a {\em minimum} node/edge multicut/multiway cut are all intractable, we rather enumerate {\em minimal} edge/node multicuts/multiway cuts instead.
    We say that a node/edge multicut $M$ of $(G, B)$ is minimal if $M'$ is not a node/edge multiway cut of $(G, B)$ for every proper subset $M' \subset M$, respectively.
    Minimal node/edge multiway cuts are defined accordingly.
    Although finding a minimal node/edge multicut is easy, our goal is to enumerate {\em all} the minimal edge/node multicuts/multiway cuts of a given graph $G$ and terminal pairs $B$.
    In this context, there are several results related to our problems.
    
    There are linear delay algorithms for enumerating all minimal $s$-$t$ (edge) cuts \cite{Provan::1994,Tsukiyama::1980}, which is indeed a special case of our problems, where $T$ contains exactly two vertices $s$ and $t$. 
    Here, an enumeration algorithm has {\em delay} complexity $f(n)$ if the algorithm outputs all the solutions without duplication and for each pair of consecutive two outputs (including preprocessing and postprocessing), the running time between them is upper bounded by $f(n)$, where $n$ is the size of the input. 
    For the node case, the problem of enumerating all minimal $s$-$t$ (node) separators has received a lot of attention and numerous efforts have been done for developing efficient algorithms \cite{Kloks::1998,Shen::1997,Takata::2010} due to many applications in several fields \cite{Bouchitte:2002,Feng::2014,Fomin:2015}.
    The best known enumeration algorithm for minimal $s$-$t$ separators was given by Tanaka \cite{Takata::2010}, which runs in $O(nm)$ delay and $O(n)$ space, where $n$ and $m$ are the number of vertices and edges of an input graph, respectively.
    
    Khachiyan et al.~\cite{Khachiyan::2008} studied the minimal edge multicut enumeration problem.
    They gave an efficient algorithm for this problem, which runs in {\em incremental polynomial time}~\cite{Johnson::1988}, that is, 
    if $\mathcal M$ is a set of minimal edge multicuts of $(G, B)$ that are generated so far, 
    then the algorithm decides whether there is a minimal edge multicut of $G$ not included in $\mathcal M$
    in time polynomial in $|V| + |E| + |\mathcal M|$.
    Moreover, if such a minimal edge multicut exists, the algorithm outputs one of them within the same running time bound.
    As we will discuss in the next section, this problem is a special case of the node counterpart and indeed a generalization of the minimal edge multiway cut enumeration problem.
    Therefore, this algorithm also works for enumerating all minimal edge multiway cuts.
    However, there can be exponentially many minimal edge multicuts in a graph.
    Hence, the delay of their algorithm cannot be upper bounded by a polynomial in terms of input size.
    To the best of our knowledge, there is no known non-trivial enumeration algorithm for minimal node multiway cuts.
    
    Let $(G = (V, E), B)$ be an instance of our enumeration problems. In this paper, we give polynomial delay or incremental polynomial delay algorithms.
    \begin{theorem}\label{thm:exp-space}
        There is an algorithm enumerates all the minimal node and edge multiway cuts of $(G, B)$ in $\order{|T(B)|\cdot |V| \cdot |E|}$ and $\order{|T(B)| \cdot |V| \cdot |E|^2}$ delay, respectively.
    \end{theorem}
    
    The algorithm in Theorem~\ref{thm:exp-space} requires exponential space to avoid redundant outputs.
    For the edge case, we can simultaneously improve the time and space consumption.
    \begin{theorem}\label{thm:poly-space}
        There is an algorithm enumerates all the minimal edge multiway cuts of $(G, B)$ in $\order{|T(B)| \cdot |V|\cdot|E|}$ delay in polynomial space.
    \end{theorem}
    
    For the most general problem among them (i.e., the minimal node multicut enumeration problem), we give an incremental polynomial time algorithm.
    \begin{theorem}\label{thm:incremental}
        There is an algorithm of finding, given a set of minimal node multicuts $\mathcal M$ of $(G, B)$, a minimal node multicut $M$ of $(G, B)$ with $M \notin \mathcal M$ if it exists and runs in time $\order{|\mathcal M|\cdot poly(n)}$.
    \end{theorem}
    
    The first and second results simultaneously improve the previous incremental polynomial running time bound obtained by applying the algorithm of Khachiyan et al.~\cite{Khachiyan::2008} to the edge multiway cut enumeration and extends enumeration algorithms for minimal $s$-$t$ cuts~\cite{Provan::1994,Tsukiyama::1980} and minimal $a$-$b$ separators \cite{Takata::2010}\footnote{However, our algorithm requires exponential space for minimal node multiway cuts, whereas Takata's algorithm \cite{Takata::2010} runs in polynomial space.}.
    The third result extends the algorithm of Khachiyan et al. to the node case.
    Since enumerating minimal node multicuts is at least as hard as enumerating minimal node multiway cuts and enumerating minimal node multiway cuts is at least as hard as enumerating minimal edge multiway cuts, this hierarchy directly reflects on the running time of our algorithms.
    
    The basic idea behind these results is that we rather enumerate a particular collection of partitions/disjoint subsets of $V$ than directly enumerating minimal edge/node multicuts/multiway cuts of $(G, B)$.
    It is known that an $s$-$t$ edge cut of $G$ is minimal if and only if the bipartition $(V_1, V_2)$ naturally defined from the $s$-$t$ cut induces connected subgraphs of $G$, that is, $G[V_1]$ and $G[V_2]$ are connected~\cite{Diestel::2012}.
    For minimal $a$-$b$ separators, a similar characterization is known using full components (see \cite{Golumbic:2004}, for example).
    These facts are highly exploited in enumerating minimal $s$-$t$ cuts \cite{Provan::1994,Tsukiyama::1980} or minimal $a$-$b$ separators \cite{Takata::2010}, and can be extended for our cases (See Sections~\ref{sec:multicut}, \ref{sec:node}, and \ref{sec:edge}).
    To enumerate such a collection of partitions/disjoint subsets of $V$ in the claimed running time, we use three representative techniques: the {\em proximity search paradigm} due to Conte and Uno~\cite{Conte::2019} for the exponential space enumeration of minimal node multiway cuts, the {\em reverse search paradigm} due to Avis and Fukuda~\cite{Avis::1996} for polynomial space enumeration of minimal edge multiway cuts, and the {\em supergraph approach}, which is appeared implicitly and explicitly in the literature~\cite{Cohen::2008,Conte::2019,Khachiyan::2008,Schwikowski::2002}, for the incremental polynomial time enumeration of minimal node or edge multicuts.
    These approaches basically define a (directed) graph on the set of solutions we want to enumerate.
    If we appropriately define some adjacency relation among the vertices (i.e. the set of solutions) so that the graph is (strongly) connected, then we can enumerate all solutions from a specific or arbitrary solution without any duplication by traversing this (directed) graph.
    The key to designing the algorithms in Theorem~\ref{thm:exp-space} and \ref{thm:poly-space} is to ensure that every vertex in the graphs defined on the solutions has a polynomial number of neighbors.

    We also consider a generalization of the minimal node multicut enumeration, which we call the minimal Steiner node multicut enumeration.
    We show that this problem is at least as hard as the minimal transversal enumeration on hypergraphs.

\section{Preliminaries}\label{sec:preli}

In this paper, we assume that a graph $G = (V, E)$ is connected and has no self-loops and no parallel edges.
Let $X \subseteq V$. We denote by $G[X]$ the subgraph of $G$ induced by $X$.
The neighbor set of $X$ is denoted by $N_G(X)$ (i.e. $N_G(X) = \{y \in V \setminus X: x \in X \land \{x, y\} \in E \}$ and the closed neighbor set of $X$ is denoted by $N_G[X] = N(X) \cup X$.
When $X$ consists of a single vertex $v$, we simply write $N_G(v)$ and $N_G[v]$ instead of $N_G(\{v\})$ and $N_G[\{v\}]$, respectively.
If there is no risk of confusion, we may drop the subscript $G$.
For a set of vertices $U \subseteq V$ (resp. edges $F \subseteq E$), the graph obtained from $G$ by deleting $U$ (resp. $F$) is denoted by $G - U$ (resp. $G - F$).

Let $B$ be a set of pairs of vertices in $V$.
We denote by $T(B) = \{s, t: \set{s, t} \in B\}$.
A vertex in $T(B)$ is called a \emph{terminal}, a pair in $B$ is called a set of \emph{terminal pairs}, and $T(B)$ is called a \emph{terminal set} or \emph{terminals}. 
When no confusion can arise, we may simply use $T$ to denote the terminal set.
A set of edges $M \subseteq E$ is an \emph{edge multicut} of $(G, B)$ if no pair of terminals in $B$ is connected in $G - M$.
When $B$ is clear from the context, we simply call $M$ an edge multicut of $G$.
An edge multicut $M$ is \emph{minimal} if every proper subset $M' \subset M$ is not an edge multicut of $G$.
Note that this condition is equivalent to that $M \setminus \{e\}$ is not an edge multicut of $G$ for any $e \in M$.  
Analogously, a set of vertices $X \subseteq V \setminus T$ is a \emph{node multicut} of $G$ if there is no paths between any terminal pair of $B$ in $G - X$. The minimality for node multicuts is defined accordingly.

The {\em demand graph for $B$} is a graph defined on $T(B)$ in which two vertices $s$ and $t$ are adjacent to each other if and only if $\set{s, t} \in B$.
When $B$ contains a terminal pair $\set{s, t}$ for any distinct $s, t \in T(B)$, that is, the demand graph for $B$ is a complete graph, a node/edge multicut is called a {\em node/edge multiway cut} of $G$.

Let $G = (V, E)$ be a graph and let $B$ be a set of terminal pairs.
The graph $G'$ obtained from the line graph of $G$ by adding a terminal $t'$ for each $t \in T$ and making $t'$ adjacent to each vertex corresponding to an edge incident to $t$ in $G$.
\begin{proposition}\label{prop:reduction}
    Let $M \subseteq E$. Then, $M$ is an edge multicut of $G$ if and only if $M$ is a node multicut of $G'$. 
\end{proposition}

    

By Proposition~\ref{prop:reduction}, designing an enumeration algorithm for minimal node multicuts/multiway cuts, it allows us to enumerate minimal edge multicuts/multiway cuts as well.
However, the converse does not hold in general.

\section{Incremental polynomial time enumeration of minimal node multicuts}\label{sec:multicut}

In this section, we design an incremental polynomial time enumeration algorithm for minimal node multicuts. 
Let $G = (V, E)$ and let $B$ be a set of terminal pairs.

For a (not necessarily minimal) node multicut $M$ of $G$, there are connected components $C_1, C_2, \ldots, C_\ell$ in $G - M$ such that each component contains at least one terminal but no component has a terminal pair in $B$.
Note that there can be components of $G - M$ not including in $\{C_1, \cdots, C_\ell\}$.
The following lemma characterizes the minimality of node multicut in this way.

\begin{lemma}
\label{lem:node:multicut}
    A node multicut $M \subseteq V \setminus T$ of $G$ is minimal if and only if there are $\ell$
    connected components $C_1, C_2, \ldots, C_{\ell}$ in $G - M$, each of which includes at least one terminal of $T$, such that    
    (1) there is no component which includes both vertices in a terminal pair and
    (2) for any $v \in M$, there is a terminal pair $(s_i, t_i)$ such that both components including $s_i$ and $t_i$ have a neighbor of $v$.
\end{lemma}

\begin{proof}
    Suppose that $M$ is a minimal node multicut of $G$.
    For each $s_i$, we let $C^s_i$ be the connected component of $G - M$ containing $s_i$ and for each $t_i$, let $C^t_i$ be the connected component of $G - M$ containing $t_i$.
    Note that some components $C^s_i$ and $C^t_j$ may not be distinct.
    Define the set of $\ell$ components $\{C_1, \ldots, C_\ell\}$ as $\{C^s_i, C^t_i : 1 \le i \le k\}$.
    Since $M$ is a multiway cut, $s_i$ and $t_i$ are contained in distinct components for every $1 \le i \le k$.
    If there is a vertex $v \in M$ such that for every terminal pair $(s_i, t_i)$, at least one of $C^s_i \cap N(v)$ and $C^t_i \cap N(v)$ is empty, then we can remove $v$ from $M$ without introducing a path between some terminal pair, which contradicts to the minimality of $M$.
    Therefore, the set of components satisfies both (1) and (2).
    
    Suppose for the converse that components $C_1, \ldots, C_\ell$ satisfy conditions (1) and (2).
    Since every $v \in M$ has a neighbor in some components $C_i$ and $C_j$ such that $C_i$ and $C_j$ respectively contain $s$ and $t$ for some terminal pair $(s, t) \in B$.
    Then, $G[V \setminus (M \setminus \{v\})]$ has a path between $s$ and $t$, which implies that $M$ is a minimal node multicut of $G$.
\end{proof}

From a minimal node multicut $M$ of $G$, we can uniquely determine the set $\mathcal C$ of $\ell$ components satisfying the conditions in Lemma~\ref{lem:node:multicut}, and vice-versa.
Given this, we denote by $\mathcal C_{M}$ the set of components corresponding to a minimal multicut $M$.
From now on, we may interchangeably use $M \subseteq V \setminus T$ and $\mathcal C_{M}$ as a minimal node multicut of $G$.
For a (not necessarily minimal) node multicut $M$ of $G$, we also use $\mathcal C_M$ to denote the set of connected components $\{C_1, \ldots, C_\ell\}$ of $G - M$ such that each component contains at least one terminal but no component has a terminal pair.

\DontPrintSemicolon
\begin{algorithm}[t]
    \caption{Traversing a solution graph $\mathcal G$ using a breadth-first search. }
    \label{algo:traversal}
    \Procedure{{\tt Traversal}($\mathcal G$)}{
        $S \gets$ an arbitrary solution\;
        $\mathcal Q, \mathcal U \gets \{S\}, \emptyset$\;
        \While{$\mathcal Q \neq \emptyset$}{
            Let $S$ be a solution in $\mathcal Q$\;
            Output $S$\tcc*{We do not output here for minimal node multicuts}
            Delete $S$ from $\mathcal Q$\;
            \For{$S' \in {\tt Neighborhood}(S, \mathcal U)$}{
                \lIf{$S' \notin \mathcal U$}{
                    $\mathcal Q, \mathcal U \gets \mathcal Q \cup \{S'\}, \mathcal U \cup \{S'\}$
                }
            }
        }
    }
\end{algorithm}

We enumerate all the minimal node multicuts of $G$ using the supergraph approach~\cite{Cohen::2008,Conte::2019,Khachiyan::2008,Khachiyan::2006}.
To this end, we define a directed graph on the set of all the minimal node multicuts of $G$, which we call a {\em solution graph}. 
The outline of the supergraph approach is described in Algorithm~\ref{algo:traversal}. 
The following ``distance'' function plays a vital role for our enumeration algorithm: For (not necessarily minimal) node multicuts $M$ and $M'$ of $G$, 
\[
    \dist{M}{M'} = \sum_{C' \in \mathcal C_{M'}} \size{C' \setminus \mcc{C'}{M}},
\]
where $\mcc{C'}{M}$ is the component $C$ of $G - M$ minimizing $|C' \setminus C|$.
If there are two or more components $C$ minimizing $|C' \setminus C|$, we define $\mcc{C'}{M}$ as the one having a smallest vertex with respect to some prescribed order on $V$ among those components.
It should be mentioned that the function {\tt dist} is not the actual distance in the solution graph which we will define later.
Note moreover that this value can be defined between two non-minimal node multicuts as $\mathcal
C_M$ is well-defined for every node multicut $M$ of $G$. 
Let $M$, $M'$, and $M''$ be (not necessarily minimal) node multicuts of $G$.
Then, we say that $M$ is {\em closer than $M'$ to $M''$} if $\dist{M}{M''} < \dist{M'}{M''}$.

\begin{lemma}
\label{lem:multicut:identity}
    Let $M$ and $M'$ be minimal node multicuts of $G$.
    Then, $M$ is equal to $M'$ if and only if $\dist{M}{M'} = 0$. 
\end{lemma}

\begin{proof}
    If $M = M'$, then $\dist{M}{M'}$ is obviously equal to zero. 
    Conversely, suppose $\dist{M}{M'} = 0$.
    Then, for every $C' \in \mathcal C_{M'}$, $C'$ is entirely contained in a component $C$ in $G - M$. 
    Let $C' \in \mathcal C_{M'}$ and let $C$ be the component of $G - M$ with $C' \subseteq C$.
    Suppose for the contradiction that there is a vertex $v$ in $C \setminus C'$.
    Since $G[C]$ is connected, we can choose $v$ so that it has a neighbor in $C'$.
    Then, $v$ belongs to $M'$.
    By Lemma~\ref{lem:node:multicut}, there are two components $C_1'$ and $C_2'$ in $\mathcal C_{M'}$ such that $C_1'$ and $C_2'$ respectively have terminals $s$ and $t$ with $\{s, t\} \in B$ and $v$ has a neighbor in both $C_1'$ and $C_2'$.
    Since $C'_1$ and $C'_2$ are contained in some components $C_1$ and $C_2$ of $G - M$, respectively, there is a path between $s$ and $t$ in $G - M$, a contradiction.
\end{proof}

From a node multicut $M$ of $G$, a function $\comp{}$ maps $M$ to an arbitrary minimal node multicut $\comp{M} \subseteq M$.
Clearly, this function computes a minimal node multicut of $G$ in polynomial time. 

\begin{lemma}
\label{lem:multicut:minimize}
    Let $M$ be a node multicut of $G$ and $M'$ a minimal node multicut of $G$. 
    Then, $\dist{\comp{M}}{M'} \le \dist{M}{M'}$ holds. 
\end{lemma}

\begin{proof}
    Since $\comp{M} \subseteq M$, it follows that every component of $G - M$ is contained in some component of $G - \comp{M}$.
    Therefore, $|C'\setminus \mcc{C'}{M}| \ge |C'\setminus \mcc{C'}{\comp{M}}|$ for every $C' \in \mathcal C_{M'}$
\end{proof}

To complete the description of Algorithm~\ref{algo:traversal}, we need to define the neighborhood of each minimal node multicut of $G$. 
We have to be take into consideration that the solution graph is strongly connected for enumerating all the minimal node multicuts of $G$.
To do this, we exploit {\tt dist} as follows.
Let $M$ and $M'$ be distinct minimal node multicuts of $G$.  
We will define the neighborhood of $M$ in such a way that it contains at least one minimal node multiway cut $M''$ of $G$ that is closer than $M$ to $M'$.
This allows to eventually have $M'$ from $M$ with Algorithm~\ref{algo:traversal}.
The main difficulty is that the neighborhood of $M$ contains such $M''$ for every $M'$, which will be described in the rest of this section.

To make the discussion simpler, we use the following two propositions. 
Here, for an edge $e$ of $G$, we let $G / e$ denote the graph obtained from $G$ by contracting edge $e$.
We use $v_e$ to denote the newly introduced vertex in $G / e$.

\begin{proposition}\label{prop:multicut:reduce:adjacent}
    Let $t_1$ be a terminal adjacent to another terminal $t_2$ in $G$.
    Suppose $\{t_1, t_2\}$ is not included in $B$. 
    Then, $M$ is a minimal node multicut of $(G, B)$ if and only if 
    it is a minimal node multicut of $(G / e, B')$, where $e = \{t_1, t_2\}$ and $B'$ is obtained by replacing $t_1$ and $t_2$ in $B$ with the new vertex $v_e$ in $G / e$.
\end{proposition}

If $G$ has an adjacent terminal pair in $B$, then obviously there is no node multicut of $G$.
By Proposition~\ref{prop:multicut:reduce:adjacent}, $G$ has no any adjacent terminals.

\begin{proposition}\label{prop:multicut:reduce:node}
    If there is a vertex $v$ of $G$ such that $N(v)$ contains a terminal pair $\set{s, t}$, 
    then for every node multicut $M$ of $G$, we have $v \in M$. 
    Moreover, $M$ is a minimal node multicut of $(G, B)$ if and only if $M \setminus \set{v}$ is a minimal node multicut of $(G - \set{v}, B)$. 
\end{proposition}

From the above two propositions, we assume that there is no pair of adjacent terminals and no vertex including a terminal pair in $B$ as its neighborhood. 
To define the neighborhood of $M$, we distinguish two cases.

\begin{figure}[t]
    \centering
    \includegraphics[width=0.7\textwidth]{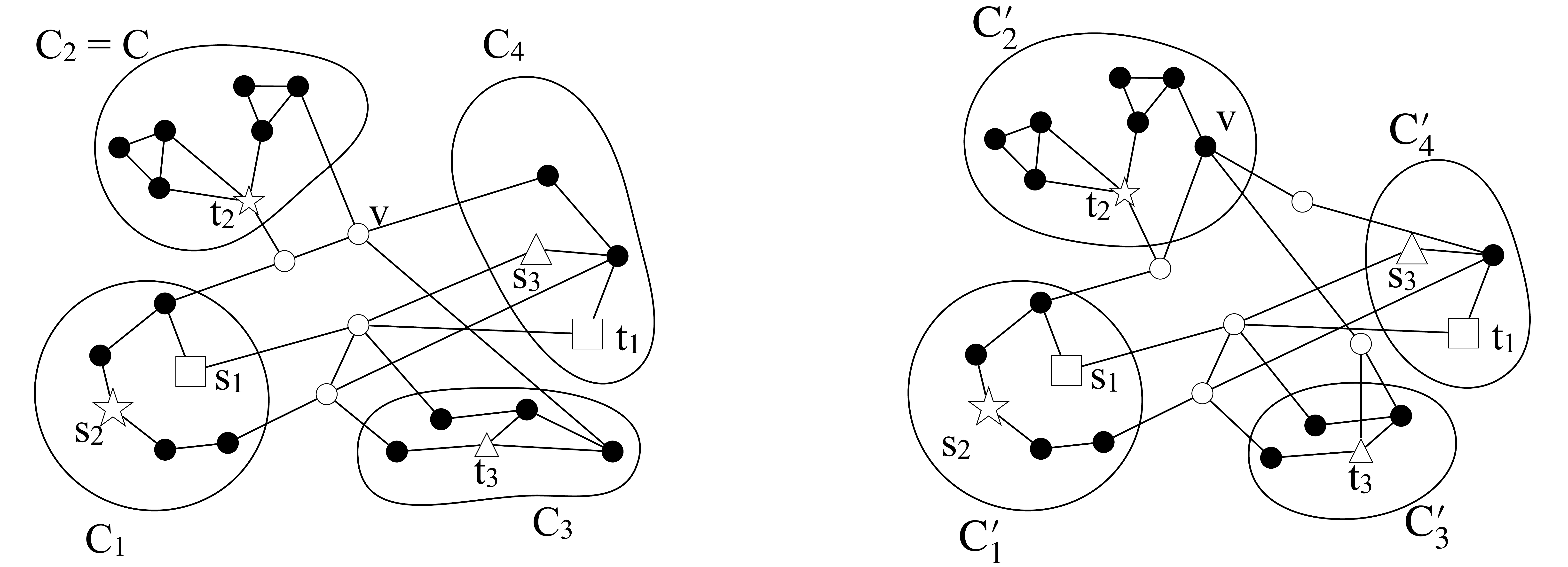}
    \caption{This figure illustrates an example of Lemma~\ref{lem:multicut:decrease:noteminal}. White circles represent vertices in node multicuts and pairs of stars, squares, and triangles represent terminal pairs.
    The left and right pictures depict a minimal node multicut $M$ and a node multicut $M''$.}\label{fig:nodemulticut1}
\end{figure}

\begin{lemma}\label{lem:multicut:decrease:noteminal}
    Let $M$ and $M'$ be distinct minimal node multicuts of $G$, let $C' \in \mathcal C_{M'}$, and let $C = \mcc{C'}{M}$.
    Suppose there is a vertex $v \in N(C) \cap C' \subseteq M$ such that $G[N[v] \cup C]$ has no any terminal pair in $B$.
    Let $T_v = N(v) \cap T$.
    Then $M'' = (M \setminus \{v\}) \cup (N(T_v \cup \{v\}) \setminus C)$ is a node multicut of $G$.
    Moreover, $\comp{M''}$ is closer than $M$ to $M'$.
\end{lemma}

\begin{proof}
    First, we show that $M''$ is a node multicut of $G$.
    To see this, suppose that there is a path in $G - M''$ between a terminal pair $\{s, t\} \in B$.
    Since $M$ is a node multicut and $M \setminus \{v\} \subseteq M''$, such an $s$-$t$ path must pass through $v$.
    As $N(T_v \cup \{v\}) \setminus C \subseteq M''$, $s$ and $t$ are contained in $T_v \cup \set{v} \cup C$, which contradicts to the fact that $G[N[v] \cup C]$ has no terminal pair in $B$.

    Next, we show that $M''$ is closer than $M$ to $M'$. 
    Recall that $v$ is contained in $N(C) \cap C'$. 
    Since $M''$ does not contain any vertex of $C$, there is a component $C''$ of $G - M''$ that contains $C \cap C'$.
    Thus, we have $C' \setminus \mcc{C'}{M''} \subseteq C' \setminus \mcc{C'}{M}$. 
    To prove that $M''$ is closer than $M$ to $M'$, it suffices to show that $\size{D' \setminus \mcc{D'}{M''}} \le \size{D' \setminus \mcc{D'}{M}}$ for each $D' \in \mathcal C_{M'} \setminus \{C'\}$. 
    To this end, we show that $G - M''$ has a component including $D' \cap \mcc{D'}{M}$, which implies $D' \cap \mcc{D'}{M} \subseteq D' \cap \mcc{D'}{M''}$.
    Let $D = \mcc{D'}{M}$.
    Observe that $D' \cap N[T_v \cup \{v\}] = \emptyset$.
    This follows from the facts that $v \in C'$ and $T_v \cap M' = \emptyset$.
    By the construction of $M''$, there is a component $D''$ in $G - M''$ with $D \setminus N[T_v \cup \{v\}] \subseteq D''$.
    Note that $D \setminus N[T_v \cup \{v\}]$ can be empty.
    In this case, as $D \subseteq N[T_v \cup \{v\}]$ and $D' \cap N[T_v \cup \{v\}]$, we have $D \cap D' = \emptyset$.
    Thus, such a component $D''$ entirely contains $D \cap D'$.
    
    Therefore, $M''$ is closer than $M$ to $M'$ and hence by Lemma~\ref{lem:multicut:minimize}, the lemma follows.
\end{proof}

Figure~\ref{fig:nodemulticut1} illustrates an example of $M$ and $M''$ in Lemma~\ref{lem:multicut:decrease:noteminal}.

If there is a terminal pair $\{s, t\} \in B$ in $G[C \cup N[v]]$, $M''$ defined in Lemma~\ref{lem:multicut:decrease:noteminal} is not a node multicut of $G$ since $s$ and $t$ are contained in the connected component $C \cup T_v \cup \{v\}$ of $G - M''$ (see Figure~\ref{fig:nodemulticut2}). 
In this case, we have to separate all terminal pairs in this component. 

\begin{lemma}\label{lem:multicut:decrese:terminal}
    Let $M$ and $M'$ be distinct two minimal node multicuts of $G$, let $C' \in \mathcal C_{M'}$, and let $C = \mcc{C'}{M}$.
    Suppose there is a vertex $v \in N(C) \cap C' \subseteq M$ such that $G[N[v] \cup C]$ contains some terminal pair in $B$.    
    Let $T_v = N(v) \cap T$.
    Then, $M'' = (M \setminus \set{v}) \cup (N(T_v \cup \set{v}) \setminus C) \cup (C \cap M')$ is a node multicut and $\comp{M''}$ is closer than $M$ to $M'$. 
\end{lemma}
\begin{proof}
    The first part of this lemma is similar to Lemma~\ref{lem:multicut:decrease:noteminal}.
    However, $G[N[v] \cup C]$ contains a terminal pair.
    Since $C$ contains no terminal pair in $B$, by Proposition~\ref{prop:multicut:reduce:node}, exactly one of a vertex in such a terminal pair is contained in $T_v$.
    Thus, as $M'$ is a node multicut of $G$, $C \cap M'$ separates such a pair and hence $M''$ is a node multicut of $G$.
    
    The second part of this lemma is also similar to Lemma~\ref{lem:multicut:decrease:noteminal}.
    Observe that the increment of $M''$ compared with one in Lemma~\ref{lem:multicut:decrease:noteminal} is $C \cap M'$.
    Since $C' \cap M'$ is empty, there is a component $C''$ of $G - M''$ that contains $C \cap C'$ and hence we have $C' \setminus \mcc{C'}{M''} \subseteq C' \setminus \mcc{C'}{M}$.
    Moreover, it holds that $|D' \setminus \mcc{D'}{M''}| \le |D' \setminus \mcc {D'}{M}|$ for every $D' \in C_{M} \setminus \{C'\}$.
    By Lemma~\ref{lem:multicut:minimize}, the lemma follows.
\end{proof}

Now, we formally define the neighborhood of a minimal node multicut $M$ in the solution graph.
For each component $C$ and $v \in N(C)$, the neighborhood of $M$ contains $\comp{(M \setminus \{v\}) \cup (N(T_v \cup \{v\}) \setminus C)}$ if $G[N[v] \cup C]$ has no any terminal pair in $B$ and $\comp{(M \setminus \{v\}) \cup (N(T_v \cup \{v\}) \setminus C) \cup (C \cap M')}$ otherwise.
By Lemmas~\ref{lem:multicut:decrease:noteminal} and \ref{lem:multicut:decrese:terminal}, this neighborhood relation ensures that the solution graph is strongly connected, which allows us to enumerate all the minimal node multicuts of $G$ from an arbitrary one using Algorithm~\ref{algo:traversal}.
However, there is an obstacle: We have to generate the neighborhood without knowing $M'$ for the case where $G[N[v] \cup C]$ has a terminal pair.
To this end, we show that computing the neighborhood for this case can be reduced to enumerating the minimal $a$-$b$ separators of a graph.

Suppose that $G[N[v] \cup C]$ has a terminal pair.
Let $M'$ be an arbitrary node multicut of $G$.
An important observation is that $C \cap M'$ is a node multicut of $(G[C \cup T_v \cup \{v\}], \{\{s, t\} : \{s, t\} \in B, s \in T_v, t \in C\})$. 
Since $C$ is a component of $G - M$, by Proposition~\ref{prop:multicut:reduce:node}, one of the terminals in each terminal pair contained in $G[N[v] \cup C]$ belongs to $N(v) \cap T$ and the other one belongs to $C$.
Thus, every path between those terminal pairs pass through $v$ in $G[C \cup T_v \cup \{v\}]$ and $v \notin M'$.
It implies that $C \cap  M'$ is a node multicut of $(G[C \cup \{v\}], \{\{v, t\} : \{s, t\} \in B, t \in C\})$.
Let $H = G[C \cup \{v\}]$ and let $B' = \{\{v, t\} : \{s, t\} \in B, t \in C\})$.
Moreover, if we have two distinct minimal node multicuts $M_1$ and $M_2$ of $(H, B')$, minimal node multicuts $\comp{(M \setminus \{v\}) \cup (N(T_v \cup \{v\})\setminus C) \cup (C \cap M_i)}$, for $i = 1, 2$, are distinct since function $\comp{}$ does not remove any vertex in $M_1$ and $M_2$.

Now, our strategy is to enumerate minimal node multicuts $\comp{C \cap M'}$ of $(H, B')$ for all $M'$. 
This subproblem is not easier than the original problem at first glance.
However, this instance $(H, B')$ has a special property that the demand graph for $B'$ forms a star.
From this property, we show that this problem can be reduced to the minimal $a$-$b$ separator enumeration problem. 

\begin{figure}[t]
    \centering
    \includegraphics[width=0.7\textwidth]{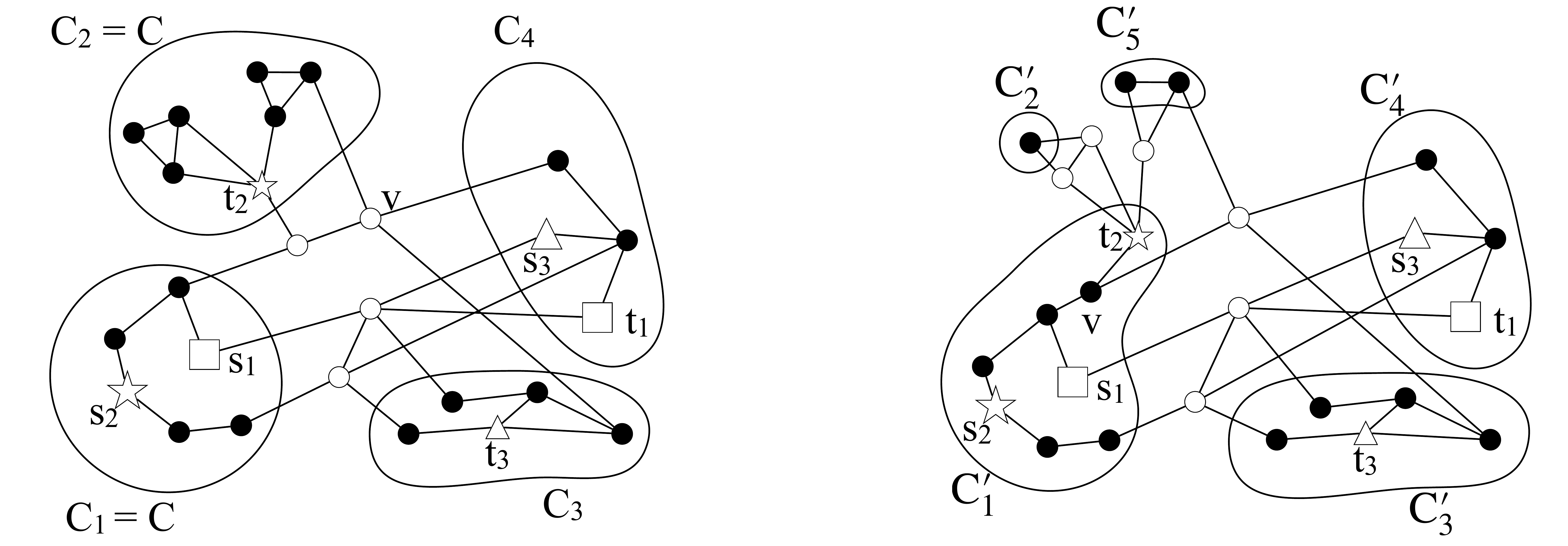}
    \caption{This figure illustrates an example of Lemma~\ref{lem:multicut:decrese:terminal}. $(M \setminus \set{v}) \cup  (N(T_v \cup \set{v}) \setminus C_1))$ is not a node multicut of $G$, and then we additionally have to separate a pair of terminals (represented by stars) in the component $C'_1 = C_1 \cup T_v \cup \{v\}$. }\label{fig:nodemulticut2}
\end{figure}

\begin{lemma}\label{lem:multicut:absep}
    Let $H = G[C \cup \{v\}]$ and let $B' = \set{\set{v, t}: \set{s,t}\in B, t \in C}$.
    Let $H'$ be the graph obtained from $H$ by identifying all the vertices of $T(B') \setminus \{v\}$ into a single vertex $v_t$.
    Then, $M \subseteq (C \cup \set{v}) \setminus T(B')$ is minimal node multicut of $(H, B')$ if and only if
    $M$ is a minimal $v$-$v_t$ separator of $H'$.
\end{lemma}

\begin{proof}
    In this proof, we use a well-known characterization of minimal separators: $M$ is a minimal $v$-$v_t$ separator of $H'$ if and only if $H' - M$ has two components $C_{v}$ and $C_{v_t}$ such that $N_{H'}(C_v) = N_{H'}(C_{v_t}) = S$, $v \in C_v$, and $v_t \in C_{v_t}$, which is a special case of Lemma~\ref{lem:node:multicut}.

    Suppose that $M$ is a minimal node multicut of $(H, B')$.
    As $M$ is minimal, $H - M$ has components $C_1, \ldots, C_\ell$ as in Lemma~\ref{lem:node:multicut}.
    Suppose that $v \in C_1$. Let $C_v = C_1$.
    Since all the vertices of $T(B') \setminus \{v\}$ are identified into $v_t$,
    $\bigcup_{2\le i\le \ell} C_i \setminus T(B')$ forms a component $C_{v_t}$ in $H'$.
    Moreover, by condition (2) of Lemma~\ref{lem:node:multicut}, $u \in M$ has a neighbor in $C_1$ and in $C_i$ for some $2 \le i \le \ell$.
    Thus, $v$ has a neighbor in $C_v$ and in $C_{v_t}$, which implies that $N_{H'}(C_v) = N_{H'}(C_{v_t}) = M$.
    
    For the converse, we suppose that $M$ is a minimal $v$-$v_t$ separator of $H'$.
    We first show that $M$ is a node multicut of $(H, B')$.
    To see this, suppose for contradiction that there is a path between $v$ and $t$ in $H - M$.
    We choose a shortest one among all $v$-$t$ paths for $t \in T(B') \setminus \{v\}$.
    Since this path has no any terminal vertex of $T(B)$ except for its end vertices $v$ and $t$, this is also a $v$-$v_t$ path in $H'$, a contradiction.
    
    To see the minimality of $M$, let $C_1, \ldots, C_\ell$ be the components of $H - M$, each of which contains at least one terminal vertex of $T(B')$.
    Suppose that $v \in C_1$.
    Let $u \in M$ be arbitrary.
    Since $M$ is a minimal $v$-$v_t$ separator of $H'$, $v$ has a neighbor in $C_1$.
    Consider a shortest $u$-$v_t$ path in $H'[C_{v_t}]$.
    Since this path has no any terminal except for $v_t$,
    it is contained in a component $C_i$ for some $2 \le i \le \ell$.
    Therefore $M$ satisfies condition (2) of Lemma~\ref{lem:node:multicut}, and hence $M$ is a minimal node multicut of $(H, B')$.
\end{proof}

By Lemma~\ref{lem:multicut:absep}, we can enumerate $\comp{C \cap M'}$ for every minimal node multicut $M'$ of $G$ by using the minimal $a$-$b$ separator enumeration algorithm of Takata~\cite{Takata::2010}.
Moreover, as observed above, for any distinct minimal $v$-$v_t$ separators $S_1$ and $S_2$ in $H'$, 
we can generate distinct minimal node multicuts $\comp{(M \setminus \{v\}) \cup (N(T_v \cup \{v\})\setminus C) \cup S_i}$ of $G$.

The algorithm generating the neighborhood of $M$ is described in Algorithm~\ref{algo:inc:enmc}.

\begin{algorithm}[t]
    \caption{Computing the neighborhood of a minimal node multicut $M$ of $(G, B)$.}
    \label{algo:inc:enmc}
    \Fn(){${\tt Neigborhood}(M, \mathcal M)$}{
        $\mathcal S \gets \emptyset$\;
        \For{$v \in M$}{
            \For{$C \in \mathcal C_M$}{
                $T_v \gets N(v) \cap T$\;
                $M'' \gets (M \setminus \set{v}) \cup (N(T_v \cup \{v\}) \setminus C))$\;
                \If{$G[C \cup N[v]]$ has no terminal pairs}{
                    \lIf{$\comp{M''} \not\in \mathcal M$}{Output $\comp{M''}$}
                    $\mathcal S \gets \mathcal S \cup \set{\comp{M''}}$\;
                }\Else{
                    Run Takata's algorithm \cite{Takata::2010} for $(H', v, v_t)$ in Lemma~\ref{lem:multicut:absep}\;
                    \ForEach{Output $M'$ of minimal $v$-$v_t$ separator in $H'$}{
                        \lIf{$\comp{M'' \cup M'} \not\in \mathcal M$}{Output $\comp{M'' \cup M'}$}
                        $\mathcal S \gets \mathcal S \cup \set{\comp{M'' \cup M'}}$\label{algo:enumab}\;
                    }
                }
            }
        }
        \Return $\mathcal S$\;
    }
\end{algorithm}

\begin{theorem}\label{theo:multicut}
    Algorithm~\ref{algo:traversal} with {\tt Neigborhood} in Algorithm~\ref{algo:inc:enmc} enumerates all the minimal multicuts of $G$ in incremental polynomial time.
\end{theorem}

\begin{proof}
    The correctness of the algorithm follows from the observation that the solution graph is strongly connected.
    Therefore, we consider the delay of the algorithm.
    
    Let $\mathcal M$ be a set of minimal node multicuts of $G$ that are generated so far.
    Let $M$ and $M'$ be arbitrary minimal node multicuts of $G$ with $M \in \mathcal M$ and $M' \notin \mathcal M$.
    By Lemmas~\ref{lem:multicut:decrease:noteminal} and \ref{lem:multicut:decrese:terminal}, Algorithm~\ref{algo:inc:enmc} finds either a minimal node multicut $\comp{M''}$ of $G$ not included in $\mathcal M$ or a minimal node multicut of $G$ that is closer than $M$ to $M'$.
    Moreover, since $\dist{M}{M'} \le n$, the algorithm outputs at least one minimal node multicut that is not contained in $\mathcal M$ in time $O(|\mathcal M|\cdot poly(n))$ if it exists.
\end{proof}

Note that, in Algorithm~\ref{algo:inc:enmc}, we use Takata's algorithm to enumerating minimal $v$-$v_t$ separators of $H'$.
To bound the delay of our algorithm, we need to process lines 13-14 for each output of Takata's algorithm.

\section{Polynomial delay enumeration of minimal node multiway cuts}\label{sec:node}

This section is devoted to designing a polynomial delay and exponential space enumeration algorithm for minimal node multiway cuts.
Let $G = (V, E)$ be a graph and let $T$ be a set of terminals.
We assume hereafter that $k = |T|$.
We begin with a characterization of minimal node multiway cuts as Lemma~\ref{lem:node:multicut}.

\begin{lemma}\label{lem:node:iff}
    A node multiway cut $M \subseteq V \setminus T$ is minimal if and only if there are $k$ connected components $C_1, C_2, \ldots, C_k$ of $V \setminus M$ such that
    (1) for each $1 \le i \le k$, $C_i$ contains $t_i$ and (2) for every $v \in M$, there is a pair of indices $1 \le i < j \le k$ with $N(v)\cap C_i \neq \emptyset$ and $N(v) \cap C_j \neq \emptyset$.
\end{lemma}

\begin{proof}
    Since every node multiway cut of $(G, T)$ is a node multicuts of $(G, T \times T)$, by Lemma~\ref{lem:node:multicut}, $M$ is a minimal node multiway cut of $(G, T)$ if and only if there are $k$ components of $G - M$, each of which, say $C_i$, has exactly one terminal $t_i$ and, for any $v \in M$, there is a pair of terminals $t_i, t_j \in T$ such that $v$ has a neighbor in $C_i$ and $C_j$, which proves the lemma.
\end{proof}

From a minimal node multiway cut $M$ of $G$, one can determine a set of $k$ connected components $C_1, \ldots, C_k$ in Lemma~\ref{lem:node:iff}.
Conversely, from a set of connected components $C_1, \ldots, C_k$ satisfying (1) and (2), one can uniquely determine a minimal node multiway cut $M$.
Given this, we denote by $\mathcal C_{M}$ a set of $k$ connected components associated to $M$.

The basic strategy to enumerate minimal node multiway cuts is the same as one used in the previous section: We define a solution graph that is strongly connected.
Let $M$ be a minimal node multiway cut of $G$ and let $\mathcal C_{M} = \{C_1, \ldots, C_k\}$.
For $1 \le i \le k$, $v \in M$ with $N(v) \cap (T \setminus \{t_i\}) = \emptyset$, let $M^{i, v} = (M \setminus \{v\}) \cup (\bigcup_{j \neq i} N(v) \cap C_j)$.
Intuitively, $M^{i,v}$ is obtained from $M$ by moving $v$ to $C_i$ and then appropriately removing vertices in $N(v)$ from $C_j$.
The key to our polynomial delay complexity is the size of the neighborhood of each $M$ is bounded by a polynomial in $n$, whereas it can be exponential in the case of minimal node multicut.

\begin{lemma}\label{lem:node:neighbor_cut}
    If $M$ is a minimal node multiway cut of $G$, then so is $M^{i, v}$.
\end{lemma}
\begin{proof}
    Suppose for contradiction that there is a path between a pair of terminals in $G - M^{i, v}$.
    Then, this path must pass through $v$ as $M$ is a node multiway cut of $G$.
    However, $M^{i, v}$ contains $\bigcup_{j \neq i} N(v) \cap C_j$, which yields a contradiction to the fact that the path connects two distinct terminals and passes through $v$.
\end{proof}

Now, we define the neighborhood of $M$ in the solution graph.
The neighborhood of $M$ consists of the set of minimal node multiway cuts $\comp{M^{i, v}}$ for every $1 \le i \le k$ and $v \in M$ with $N(v) \cap (T \setminus \{t_i\}) = \emptyset$.
To show the strong connectivity of the solution graph, we define 
\[
\dist{M}{M'} = \sum_{1 \le i \le k} \size{C'_i \setminus C_i},
\]
where $\mathcal C_M = \set{C_1, \ldots, C_k}$ and $\mathcal C_{M'} = \set{C'_1, \ldots, C'_k}$.
Note that the definition of $\tt dist$ is slightly different from one used in the previous section.
Let $M$, $M'$, $M''$ be minimal node multiway cuts of $G$.
We say that $M$ is {\em closer than $M''$ to $M$} if $\dist{M}{M'} < \dist{M''}{M'}$.

\begin{lemma}\label{lem:node:identity}
    Let $M$ and $M'$ be minimal node multiway cuts of $G$.
    Then, $\dist{M}{M'} = 0$ if and only if $M = M'$. 
\end{lemma}

\begin{proof}
    Obviously, $\dist{M}{M'} = 0$ if $M = M'$. Thus we prove the other direction.
    
    Suppose that $\dist{M}{M'} = 0$.
    Let $\mathcal C_{M} = \set{C_1, \ldots, C_k}$ and $\mathcal C_{M'} = \set{C'_1, \ldots, C'_k}$.
    Then, we have $C_i \subseteq C'_i$ for every $1 \le i \le k$.
    Suppose for a contradiction that there is $v \in C'_i \setminus C_i$.
    Since $G[C'_i]$ is connected, we can choose $v$ in such a way that it has a neighbor in $C_i$.
    Since $M''$ is a minimal multiway node cut of $G$ and $v \in C'_i$, we have $N(v) \cap C'_j = \emptyset$ for each $j \neq i$. 
    As $C_j \subseteq C'_j$, $N(v) \cap C_j = \emptyset$ holds for any $j \neq i$. 
    Thus, we have $v \notin M$, which implies that $C_i$ is not a connected component of $G - M$.
    By Lemma~\ref{lem:node:iff}, $M$ is not a minimal node multiway cut of $G$, a contradiction.
\end{proof}

\begin{lemma}\label{lem:dec:depth}
    Let $M$ and $M'$ be distinct minimal node multiway cuts of $G$.
    Then, there is a minimal node multiway cut $M''$ of $G$ in the neighborhood of $M$ such that $M''$ is closer than $M$ to $M'$.
\end{lemma}

\begin{proof}
    Let $\mathcal C_M = \{C_1, \ldots, C_k\}$ and $\mathcal C_{M'} = \{C'_1, \ldots, C'_k\}$. 
    By Lemma~\ref{lem:node:identity}, there is a pair $C_i$ and $C'_i$ with $C'_i \setminus C_i \neq \emptyset$.
    As $G[C'_i]$ is connected, there exists a vertex $v$ in $C'_i \cap N(C_i)$. 
    Note that $v \in M$ as otherwise $v$ must be contained in $C_i$.
    By the definition of neighborhood, there is a minimal node multiway cut $M''$ of $G$ with $M'' = \comp{M^{i, v}}$.
    Let $C_{M''} = \set{C''_1, \ldots, C''_k}$.
    Observe that $|C'_i \setminus C''_i| < |C'_i \setminus C_i|$. 
    This follows from the fact that $C_i \cup \{v\} \subseteq C''_i$.
    Let $j \neq i$.
    By the definition of $M^{i, v}$, it holds that $C_j \setminus N(v) \subseteq C''_j$.
    Since $C'_i$ contains $v$, $C'_j$ does not contain any vertex in $N(v)$. 
    Thus, we have $C'_j \cap C''_j \subseteq C'_j \cap C_j$, and hence
    $\size{C'_j \setminus C_j} = \size{C'_j} - \size{C'_j \cap C_j }\le \size{C'_j} - \size{C'_j \cap C''_j} =\size{C_j \setminus C'_j}$,
    which completes the proof.
\end{proof}

Similarly to the previous section, by Lemma~\ref{lem:dec:depth}, we can conclude that the solution graph is strongly connected.
From this neighborhood relation, our enumeration algorithm is quite similar to one in the previous section, which is described in Algorithm~\ref{algo:enmc}.
To bound the delay of Algorithm~\ref{algo:enmc}, we need to bound the time complexity of computing $\comp{M}$.

\begin{lemma}\label{lem:comp:time}
    Let $M$ be a node multiway cut of $G$. 
    Then, we can compute $\comp{M}$ in $\order{n + m}$ time.
\end{lemma}

\begin{proof}
    We first compute the set of connected components of $G - M$.
    Let $C_i$ be the component including $t_i$ for $1 \le i \le k$.
    We build a data structure that, given a vertex $v$, reports the index $i$ if $v \in C_i$ in constant time, using a one-dimensional array.
    These can be done in linear time.
    Now, for each $v \in M$, we check if $M \setminus \{v\}$ is a node multiway cut of $G$.
    This can be done in $O(d(v))$ time using the above data structure.
    If we remove $v$ from $M$, we have to update the data structure: Some components not in $\set{C_1,\ldots, C_k}$ are merged into $C_i$.
    Each vertex is updated at most once in computing $\comp{M}$.
    Overall, we can in linear time compute $\comp{M}$.
\end{proof}

\begin{theorem}\label{theo:enmc}
    \cref{algo:traversal} with {\tt Neigborhood} in \Cref{algo:enmc} enumerates all the minimal node multiway cuts of $G$ in $\order{knm}$ delay and exponential space. 
\end{theorem}

\begin{proof}
    The correctness of the algorithm immediately follows from Lemma~\ref{lem:dec:depth}.
    Therefore, in the following, we concentrate on running time analysis. 
    
    In the first line of Algorithm~\ref{algo:traversal}, we compute an arbitrary minimal node multiway cut of $G$ in time $O(n + m)$ using the algorithm in Lemma~\ref{lem:comp:time}.
    For each output $M$, we compute the neighborhood of $M$ and check the dictionary $\mathcal U$ whether it has already been generated.
    For $1 \le i \le k$ and $v \in M$, we can compute $M^{i, v}$ and $\comp{M^{i, v}}$ in time $O(n + m)$ by Lemma~\ref{lem:comp:time}.
    Since the neighborhood of $M$ contains at most $kn$ minimal node multiway cuts of $G$ and we can check if the dictionary contains a solution in time $O(n)$,
    the delay is $O(k(n^2 + nm)) = O(knm)$.
\end{proof}

\begin{algorithm}[t]
    \caption{Computing the neighborhood of a minimal node multiway cut $M$ of $G$.}
    \label{algo:enmc}
    \Fn(){${\tt Neighborhood}(M, \mathcal M)$}{
        $\mathcal S \gets \emptyset$\;
        \For{$v \in M$}{
            \For{$C_i \in \mathcal C_M$}{
                \lIf{$N(v) \setminus C_i$ has no terminals}{
                    $\mathcal S \gets \mathcal S \cup \comp{M^{i,v}}$
                }
            }
        }
        \Return $\mathcal S$\;
    }
\end{algorithm}



\section{Polynomial space enumeration for minimal edge multiway cuts}\label{sec:edge}
In the previous section, we have developed a polynomial delay enumeration for both node multiway cuts.
\Cref{prop:reduction} and the previous result imply that the minimal edge multiway cut enumeration problem can be solved in polynomial delay and exponential space.
In this section, we design a polynomial delay and space enumeration for minimal edge multiway cuts.
Let $G = (V, E)$ be a graph and let $T$ be a set of terminals.

\begin{lemma}
\label{lem:con}
    Let $M \subseteq E$ be an edge multiway cut of $G$.
    Then, $M$ is minimal if and only if $G - M$ has exactly $k$ connected components $C_1, \ldots, C_k$, each $C_i$ of which contains $t_i$.
\end{lemma}

\begin{proof}  
    Suppose that $M$ is a minimal edge multiway cut of $G$.
    From the definition of edge multiway cut, $G - M$ has at least $k$ connected components $C_1, C_2, \ldots, C_{k'}$. 
    We can assume without loss of generality that each $C_i$ contains $t_i$. 
    If $M$ contains an edge of $G[C_i]$ for some $i$, we can simply remove this edge from $M$ without introducing a path between terminals, which contradicts to the minimality of $M$.
    Moreover, if $k' > k$, there is at least one edge $e$ in $M$ such that one of the end vertices of $e$ belongs to $C_i$ for some $i \le k$ and the other end vertex of $e$ belongs to $C_j$ for some $j > k$.
    This edge can be removed from $M$ without introducing a path between terminals, contradicting to the minimality of $M$.
    Therefore, the ``only if'' part follows.

    Conversely, let $C_1, C_2, \ldots C_k$ be the connected components of $G - M$ such that $C_i$ contains $t_i$ for each $1 \le i \le k$.
    Every edge $e$ in $M$ lies between two connected components, say $C_i$ and $C_j$.
    This implies that there is a path between $t_i$ and $t_j$ in $G - (M \setminus \{e\})$.
    Hence, $M$ is minimal.
\end{proof}

Note that the lemma proves in fact that there is a bijection between the set of minimal multiway cuts of $G$ and the collection of partitions of $V$ satisfying the condition in the lemma.
In what follows, we also regard a minimal multiway cut $M$ of $G$ as a partition $\mathcal P_{M} = \{C_1, C_2, \ldots C_k\}$ of $V$ satisfying the condition in Lemma~\ref{lem:con}.
We write $\mathcal P_M^{i<}$, $\mathcal P_M^{<i}$, and $\mathcal P_M^{\le i}$ to denote $\bigcup_{i < j} C_j$, $\bigcup_{j < i} C_j$, and $\bigcup_{j \le i} C_j$, respectively.
For a vertex $v \in V$, 
the position of $v$ in $\mathcal P_M$, denoted by $\mathcal P_M(v)$, is the index $1 \le i \le k$ with $v \in C_i$. 

The bottleneck of the space complexity for enumeration algorithms in the previous sections is to use a dictionary to avoid duplication. 
To overcome this bottleneck, we propose an algorithm based on \emph{the reverse search paradigm}~\cite{Avis::1996}.
Fix a graph $G = (V, E)$ and a terminal set $T \subseteq V$.
In this paradigm, we also define a graph on the set of all minimal edge multiway cuts of $G$ and a specific minimal edge multiway cut, which we call the {\em root}, denoted by $R \subseteq V$.
By carefully designing the neighborhood of each minimal edge multiway cut of $G$, the solution graph induces a {\em directed tree} from the root, which enables us to enumerate those without duplication in polynomial space.


To this end, we first define the root $\mathcal P_R = \set{C^r_1, \ldots, C^r_k}$ as follows:
Let $C^r_i$ be the component in $G - (\mathcal P_R^{<i} \cup \{t_{i+1}, \ldots, t_k\})$ including $t_i$.
Note that $\mathcal P_R^{<1}$ is defined as the empty set and hence $C^r_1$ is well-defined. 
\begin{lemma}\label{lem:root}
    The root $R$ is a minimal edge multiway cut of $G$.
\end{lemma}
\begin{proof}
    Clearly, $C^r_i$ contains $t_i$ for all $1 \le i \le k$.
    Thus, we show that $\mathcal P_R$ is a partition of $V$.
    Let $v$ be an arbitrary vertex of $G$.
    Since $G$ is connected, $v$ is adjacent to a vertex in $C^r_i$ for some $1 \le i \le k$. This implies that $v$ is included in $C^r_j$ for some $j \le i$.
\end{proof}

Next, we define the parent-child relation in the solution graph. As in the previous sections, we define a certain measure for minimal edge multiway cuts $M$ of $G$: The \emph{depth} of $M$ as
\[
    \depth{M} = \sum_{v \in V} (\mathcal P_M(v) - \mathcal P_R(v)).
\]
Intuitively, the depth of $M$ is the sum of a ``difference'' of the indices of blocks in $\mathcal P_M$ and $\mathcal P_R$ that $v$ belongs to.
For two minimal edge multiway cuts $M$ and $M'$ of $G$, we say that $M$ is \emph{shallower than} $M'$ if $\depth{M} < \depth{M'}$.
Note that the depth of $M$ is at most $kn$ for minimal edge multiway cut $M$ of $G$.
One may think that the depth of $M$ or more specifically $\mathcal P_M(v) - \mathcal P_R(v)$ can be negative. The following two lemmas ensure that it is always non-negative.

\begin{lemma}\label{lem:shiftable}
    Let $M$ be a minimal edge multiway cut of $G$ and let $\mathcal P_M = \set{C_1, \ldots, C_k}$.
    Then, $C_i \subseteq \mathcal P_R^{\le i}$ holds for every $1 \le i \le k$. 
\end{lemma}
\begin{proof}
    Suppose for contradiction that $v$ is a vertex in  $C_i \setminus \mathcal P_R^{\le i}$. 
    Since $v$ is included in $C_i$, there is a path between $t_i$ to $v$ in $G - (T \setminus \set{t_i})$. 
    By the definition of $R$, $v$ is included in $\mathcal P_R^{\le i}$, which contradicts to the fact that $v$ is a vertex in $C_i \setminus \mathcal R_{\le i}$. 
\end{proof}

\begin{lemma}\label{lem:sim}
    Let $M$ be a minimal edge multiway cut of $G$ and let $\mathcal P_M = \set{C_1, \ldots, C_k}$.
    Then, $\depth{M} = 0$ if and only if $M = R$.
\end{lemma}

\begin{proof}
    Obviously, the depth of $R$ is zero.
    Thus, in the following, we consider the ``only if'' part. 
    By Lemma~\ref{lem:shiftable}, every vertex $v \in C_i$ is included in $\mathcal P_R^{\le i}$. 
    This implies that $\mathcal P_M(v) - \mathcal P_R(v)$ is non-negative. 
    Since the depth of $M$ is equal to zero, we have $\mathcal P_M(v) = \mathcal P_R(v)$ for every $v \in V$.
    Hence, we have $C_i = C^r_i$ for every $1 \le i \le k$.
\end{proof}

Let $M$ be a minimal edge multiway cut of $G$.
To ensure that the solution graph forms a tree, we define the parent of $M$ which is shallower than $M$. 
Let $\mathcal P_M = \set{C_1, \ldots, C_k}$.
We say that a vertex $v \in (N(C_i) \cap \mathcal P_M^{i<}) \setminus T$ is \emph{shiftable into $C_i$} (or simply, \emph{shiftable}). .
In words, a vertex is shiftable into $C_i$ if it is non-terminal, adjacent to a vertex in $C_i$, and included in $C_j$ for some $j > i$.

\begin{lemma}\label{lem:sp}
    Let $M$ be a minimal node multiway cut of $G$ with $M \neq R$ and let $\mathcal P_M = \set{C_1, \ldots, C_k}$.
    Then, there is at least one shiftable vertex in $V \setminus M$.
\end{lemma}

\begin{proof}
    By Lemma~\ref{lem:sim}, the depth of $M$ is more than zero. 
    This implies that there is a vertex $v \in C_j \cap C^r_i \neq \emptyset$ for some $i \neq j$. 
    Note that $v$ is not a terminal.
    By Lemma~\ref{lem:shiftable}, we have $i < j$. 
    Observe that $C^r_i \setminus C_j$ is not empty since $C^r_i$ contains terminal $t_i$ that is not contained in $C_j$. 
    Since $G[C^r_i]$ is connected, there is at least one vertex $w \in C^r_i \setminus C_j$ that is adjacent to $v$.
    If $j < \mathcal P_M(w)$, we have $w \neq t_i$ and hence $w$ is shiftable into $C_j$.
    Otherwise, $j > \mathcal P_M(w)$, we can conclude that $v$ is shiftable into $C_{\mathcal P_M(w)}$. 
    Hence the lemma follows.
\end{proof}

Let $\mathcal P_M = \set{C_1, \ldots, C_k}$ with $M \neq R$.
By Lemma~\ref{lem:sp}, $V \setminus M$ has at least one shiftable vertex.
The largest index $i$ of a component $C_i$ into which there is a shiftable vertex is denoted by $\li{M}$. 
There can be more than one vertices that are shiftable into $C_\li{M}$.
We say that a vertex $v$ is the {\em pivot} of $M$ if $v$ is shiftable into $C_\li{M}$, and moreover, if there are more than one such vertices, we select the pivot in the following algorithmic way:
\begin{enumerate}
    \item Let $Q$ be the set of vertices, each of which is shiftable into $C_{\li{M}}$.
    \item If $Q$ contains more than one vertices, we replace $Q$ as $Q := Q \cap C_s$, where $s$ is the maximum index with $Q \cap C_s \neq \emptyset$.
    \item If $Q$ contains more than one vertices, we compute the set of cut vertices of $G[C_s]$.
    If there is at least one vertex in $Q$ that is not a cut vertex of $G[C_s]$, remove all the cut vertices of $G[C_s]$ from $Q$. Otherwise, that is, $Q$ contains cut vertices only, remove a cut vertex $v \in Q$ from $Q$ if there is another cut vertex $w \in Q$ of $G[C_s]$ such that every path between $w$ and $t_s$ hits $v$.
    \item If $Q$ contains more than one vertices, remove all but arbitrary one vertex from $Q$.
\end{enumerate}
Note that if we apply this algorithm to $Q$, $Q$ contains exactly one vertex that is shiftable into $C_{\li{M}}$.
We select the remaining vertex in $Q$ as the pivot of $M$.
Now, we define the {\em parent} of $M$ for each $M \neq R$, denoted by $\Par{M}$, as follows: Let $\mathcal P_{\Par{M}} = \set{C'_1, \ldots, C'_k}$ such that
\[
    C'_i = 
    \left\{ \begin{array}{ll}
        C_i & (i \neq \li{M}, \mathcal P_M(p)) \\
        \displaystyle C_i \cup (C_{\mathcal P_M(p)} \setminus C)    & (i = \li{M})) \\
        C    & (i = \mathcal P_M(p)),
    \end{array} \right.
\]
where $p$ is the pivot of $M$ and $C$ is the component in $G[C_{\mathcal P_M(p)} \setminus \set{p}]$ including terminal $t_{\mathcal P_M(p)}$. 
Since $p$ has a neighbor in $C_{\li{M}}$, $G[C'_{\li{M}}]$ is connected, and hence $\Par{M}$ is a minimal edge multiway cut of $G$ as well.
If $M = \Par{M'}$ for some minimal edge multiway cut $M'$ of $G$, $M'$ is called a {\em child} of $M$.
The following lemma shows that $\Par{M}$ is shallower than $M$.
\begin{lemma}\label{lem:shallower}
    Let $M$ be a minimal edge multiway cut of $G$ with $M \neq R$.
    Then, $\Par{M}$ is shallower than $M$. 
\end{lemma}

\begin{proof}
    From the definition of shiftable vertex, it follows that $\mathcal P_M(p) > \li{M}$.
    This implies that $C'_i \subseteq C_i$ for $C_i \in \mathcal P_M$ and $C'_i \in \mathcal P_{\Par{M}}$.
\end{proof}

This lemma ensures that for every minimal edge multiway cut $M$ of $G$, we can eventually obtain the root $R$ by tracing their parents at most $kn$ times.

Finally, we are ready to design the neighborhood of each minimal edge multiway cut $M$ of $G$. 
The neighborhood of $M$ is defined so that it includes all the children of $M$ and whose size is polynomial in $n$.
Let $C$ be a set of vertices that induces a connected subgraph in $G$.
The \emph{boundary of $C$}, denoted by $B(C)$, is the set of vertices in $C$ that has a neighbor outside of $C$.

\begin{lemma}\label{lem:children}
    Let $M$ and $M'$ be minimal edge multiway cut of $G$ with $\Par{M'} = M$.
    Let $\mathcal P_{M} = \set{C_1, \ldots, C_k}$.
    Then, the pivot $p$ of $M'$ belongs to the boundary of $C_{\li{M'}}$ and is adjacent to a vertex in $C_{\mathcal P_{M'}(p)}$.
\end{lemma}

\begin{proof}
    Let $\mathcal P_{M'} = \set{C'_1, \ldots, C'_k}$ and let $s = \mathcal P_{M'}(p)$. 
    Since $p$ is shiftable, it belongs to the boundary of $C'_{s}$.
    Moreover, $p$ belongs to $C_{\li{\mathcal P'}}$.
    Since $G[C'_s]$ is connected and has at least two vertices ($p$ and $t_s$),
    $p$ has a neighbor $w$ in $C'_s$.
    We can choose $w$ as a vertex in the component of $G[C'_s \setminus \{p\}]$ including terminal $t_s$.
    This implies that $\mathcal P_M(w) = s$ and hence $p$ belongs to the boundary of $C_{\li{M'}}$
    and is adjacent to $w \in C_s$.
\end{proof}

The above lemma implies every pivot of a child of $M$ is contained in a boundary of $C_i$ for some $1 \le i \le k$.
Thus, we define the neighborhood of $M$ as follows.
Let $\mathcal P_M = \{C_1, \ldots, C_k\}$.
For each $C_i$, we pick a vertex $v \in B(C_i)$ with $v \neq t_i$. Let $C$ be the set of components in $G[C_i \setminus \set{v}]$ which does not include $t_i$.
Note that $C$ can be empty when $v$ is not a cut vertex in $G[C_i]$. 
For each $1 \le i < j \le k$ and $N(v) \cap C_j \neq \emptyset$,
$\mathcal P_{M'} = \set{C'_1, \ldots, C'_k}$ is defined as:
\[
    C'_\ell = 
    \left\{ \begin{array}{ll}
        C_\ell & (\ell \neq i, j) \\
        C_\ell \cup (C \cup \set{v}) & (\ell = j) \\
        C_\ell \setminus (C \cup \set{v}) & (\ell = i).
    \end{array} \right.
\]
The neighborhood of $M$ contains such $M'$ if $\Par{M'} = M$ for each choice of $C_i$, $v \in B(C_i) \setminus \set{t_i}$, and $C_j$.
The heart of our algorithm is the following lemma.

\begin{algorithm}[t]
    \caption{Enumerating the minimal multiway cuts of $G$ in $\order{knm}$ delay and $\order{kn^2}$ space. }
    \label{algo:emc}
    \Procedure{\EnumMC{$G, M, d$}}{
        \lIf{$d$ is even}{Output $M$}
        \For{$C_i \in \mathcal P_M$\label{algo:loop1}}{
            \For{$v \in B(C_i)$ with $v \neq t_i$\label{algo:loop2}}{
                $\mathcal P' \gets \mathcal P$\tcp*{$\mathcal P' = \set{C'_1, \ldots, C'_k}$}
                $C'_i \gets$ the component including $t_i$ in $G[C_i \setminus \set{v}]$\;
                $C \gets C_i \setminus C'_i$\;
                \For{$j$ with $j > i$ and $N(v) \cap C_j \neq \emptyset$\label{algo:loop3}}{
                    $C'_j \gets C_j \cup C$\; 
                    \lIf{$\Par{M'} = M$\label{algo:par}}{
                        \EnumMC{$G, M', d + 1$}
                    }
                    $C'_j \gets C_j$\label{algo:loop3end}\; 
                }
                \label{algo:loop2end}
            }
        }
        \label{algo:loop1end}
        \lIf{$d$ is odd}{Output $M$}
    }
\end{algorithm}

\begin{lemma}\label{lem:correctness}
    Let $M$ be a minimal edge multiway cut of $G$.
    Then, the neighborhood of $M$ includes all the children of $M$.
\end{lemma}

To prove this lemma, we first show the following technical claim.
\begin{claim*}\label{lem:pivot}
    Let $p$ be the pivot of $M$ and let $s = \mathcal P_M(p)$. Then, for every connected component $C$ of $G[C_{s} \setminus \{p\}]$, either $C$ contains terminal $t_s$ or $C$ has no any shiftable vertex into $C_{\li{M}}$.
\end{claim*}
\begin{proof}[Proof of Claim]
    If $p$ is not a cut vertex in $G[C_s]$, clearly $G[C_s \setminus \{p\}]$ has exactly one component, which indeed has terminal $t_s$.
    Suppose otherwise. If there is a component $C$ of $G[C_s \setminus \{p\}]$ that has no terminal $t_s$ and has a shiftable vertex $v$ into $C_{\li{M}}$. 
    By the definition of $p$, $v$ is also a cut vertex of $G[C_s]$.
    Then, every path between $v$ and $t_s$ hits $p$. This contradicts to the choice of $p$. 
\end{proof}

\begin{proof}[Proof of Lemma~\ref{lem:correctness}]
    Let $M'$ be an arbitrary children of $M$ and let $\mathcal P_{M} = \set{C_1, \ldots, C_k}$ and $\mathcal P_{M'} = \set{C'_1, \ldots, C'_k}$. 
    By the definition of parent, every component $C_i$ except two is equal to the corresponding component $C'_i$.
    The only difference between them is two pairs of components $(C_{\li{M'}}, C'_{\li{M'}})$ and $(C_{\mathcal P_{M'}(p)}, C'_{\mathcal P_{M'}(p)})$. 
    
    Recall that, in constructing the neighborhood of $M$, we select a component $C_i$, $v \in B(C_i)$ with $v \neq t_i$, and a component $C_j$ with $N(v) \cap C_j \neq \emptyset$.
    By Lemma~\ref{lem:children}, the pivot $p$ of $M'$ is included in the boundary of $C_{\li{M'}}$.
    Moreover, since, by Lemma~\ref{lem:children}, $p$ has a neighbor in $C'_{\mathcal P'(p)}$, 
    Thus, we can correctly select $i = \li{M'}$, $v = p$, and $j = \mathcal P_{M'}(p)$.
    
    Now, consider two components $C_{i}$ and $C_{j}$.
    By the definition of parent, $C_{i} = C'_{i} \cup (C'_{j} \setminus C_{j})$ and $C_{j}$ is the component of $G[C'_j \setminus \{v\}]$ including terminal $t_j$.
    Since $C'_{i} \cap C'_j = \emptyset$ and $C_{i} \cap C_{j} = \emptyset$, we have $C'_{i} = C_{i} \setminus (C'_j \setminus C_j)$.
    By the above claim, $C'_j \setminus C_j$ has only one shiftable vertex into $C'_{i}$, which is the pivot $v$ of $M$.
    By the definition of shiftable vertex, there are no edges between a vertex in $C'_j \setminus C_j \setminus \{v\}$ and a vertex in $C'_i$.
    This means that either $C'_j \setminus C_j \setminus \{v\}$ is empty or $v$ is a cut vertex in $G[C_{i} \setminus \{v\}]$ that separates $C'_j \setminus C_j \setminus \{v\}$ from $C'_{i}$.
    Therefore, $C'_{i}$ is the component of $G[C'_i \setminus \{v\}]$ including terminal $t_{i}$. 
    Moreover, by the definition of parent, we have $C'_j = C_j \cup (C_{i} \setminus C'_{i})$.
    Hence, the statement holds.
\end{proof}

Based on Lemma~\ref{lem:correctness}, Algorithm~\ref{algo:emc} enumerates all the minimal edge multiway cuts of $G$. Finally, we analyze the delay and the space complexity of this algorithm. To bound the delay, we use the alternative output method due to Uno~\cite{Uno::2003}.

\begin{theorem}\label{theorem:runtime}
    Let $G$ be a graph and $T$ be a set of terminals. 
    Algorithm~\ref{algo:emc} runs in $\order{knm}$ delay and $\order{kn^2}$ space, where $n$ is the number of vertices, $m$ is the number of edges, and $k$ is the number of terminals. 
\end{theorem}

\begin{proof}
    Let $\mathcal T$ be the solution graph for minimal edge multiway cuts of $G$.
    First, we analyze the total running time and then prove the delay bound.
    
    Let $M$ be a minimal edge multiway cut of $G$. 
    In each node of $\mathcal T$, line~\ref{algo:loop1} guesses the component $C_i \in \mathcal P$ and line~\ref{algo:loop2} guesses the vertex in the boundary $B(C_i)$. 
    The loop block from line~\ref{algo:loop2} to line~\ref{algo:loop2end} is executed at most $n$ times in total since the total size of boundaries is at most $n$. 
    The computation of $\mathcal P_{M'}$ and $\Par{M'}$ can be done in $\order{m}$ time for each $j$ by keeping $\mathcal P_M$ with $M$.
    Thus, each node of $\mathcal T$ is processed in $\order{knm}$ time.
    Since the algorithm outputs exactly one minimal edge multiway cut of $G$ in each node of $\mathcal T$, the total computational time is $\order{knm\size{\mathcal M}}$, where $\mathcal M$ is the set of minimal edge multiway cuts of $G$.
    Moreover, we can bound the delay in $\order{knm}$ time using the alternative output method~\cite{Uno::2003} since this algorithm outputs a solution in each node in $\mathcal T$. 
    The detailed discussion is postponed to the last part of the proof.
    
    We show the space complexity bound. 
    Let $M$ be a minimal edge multiway cut of $G$ and $P$ be a path between $M$ and the root $R$ in $\mathcal T$. 
    In each node, we need to store $\mathcal P_{M'}$ and the boundary for each $C'_i \in \mathcal P_{M'}$.
    Since the size of each set in $\mathcal P_{M'}$ is $\order{n}$ and the depth of $\mathcal T$ is $kn$, the space complexity is $\order{kn^2}$. 
    
    To show the delay bound, we use the alternative output method due to~\cite{Uno::2003}. 
    We replace each edge of $\mathcal T$ with a pair of parallel edges. 
    Then, the traversal of $\mathcal T$ naturally defines an Eulerian tour on this replaced graph.
    Let $S = (n_1, \dots, n_t)$ be the sequence of nodes that appear on this tour in this order.
    Note that each leaf node appears exactly once in $S$ and each internal node appears more than once in $S$.
    From now on, we may call each $n_i$ an {\em event} and denote by $e_i$ the $i$-th event in $S$.
    Observe that if the depth of $n_i$ is even (resp. odd) in $\mathcal T$, then the first (resp. last) event in $S$ corresponding to this node outputs a solution.
    Now, let us consider three consecutive events $e_i$, $e_{i + 1}$, and $e_{i + 2}$ in $S$. 
    Since each node is processed in $O(knm)$ time, it suffices to show that at least one of these events outputs a solution.
    If at least one of these events corresponds to a leaf node, this claim obviously holds.
    Hence, we assume not in this case.
    Since each of $e_i$, $e_{i+1}$, and $e_{i+2}$ corresponds to an internal node of $\mathcal T$,
    there are three possibilities (Figure~\ref{fig:alternate}):   
    (1)   $n_{i + 1}$ is a child of $n_i$ and $n_{i + 2}$ is a child of $n_{i + 1}$. 
    (2)  $n_{i + 1}$ is a parent of $n_i$ and $n_{i + 2}$ is a parent of $n_{i + 1}$. 
    (3) $n_{i + 1}$ is a parent of both $n_i$ and $n_{i + 2}$.
    \begin{figure}
        \centering
        \includegraphics[width=0.6\textwidth]{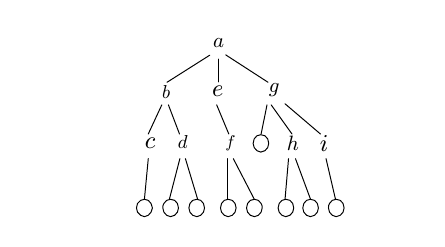}
        \caption{The figure depicts an example of the three cases of consecutive three events $e_i$, $e_{i+1}$, and $e_{i+2}$ in traversing $\mathcal T$: (1) $a,b,c$; (2) $f,e,a$; (3) $h,g,i$.}
        \label{fig:alternate}
    \end{figure}
    Note that these three nodes must be distinct since none of them is a leaf of $\mathcal T$.
    For case (1), the events $e_{i+1}$ and $e_{i+2}$ are the first events for distinct nodes $n_{i+1}$ and $n_{i+2}$, respectively.
    Since exactly one of $n_{i+1}$ and $n_{i+2}$ has even depth, therefore, either $e_{i+1}$ or $e_{i+2}$ outputs a solution.
    For case (2), the events $e_{i+1}$ and $e_{i+2}$ are the last events for those nodes, and hence exactly one of them outputs a solution as well.
    For case (3), suppose first that $n_{i + 1}$ has even depth.
    Then $n_i$ has odd depth, and hence $e_i$ is the last event for $n_i$ and hence $e_i$ outputs a solution.
    Suppose otherwise that $n_{i+1}$ has odd depth. Then, $n_{i+2}$ has even depth and $e_{i+2}$ is the first event for this node.
    This, $e_{i+2}$ outputs a solution.
    Therefore, the delay is $\order{knm}$. 
\end{proof}

\section{Minimal Steiner node multicuts enumeration}
\label{sec:steiner}
We have developed efficient enumeration algorithms for minimal multicuts and minimal multiway cuts so far.
In this section, we consider a generalized version of node multicuts, called {\em Steiner node multicuts}, and discuss a relation between this problem and the minimal transversal enumeration problem on hypergraphs.

Let $G = (V, E)$ be a graph and let $T_1, T_2, \ldots T_k \subseteq V$.
A subset $S \subseteq V \setminus (T_1 \cup T_2 \cup \cdots \cup T_k)$ is call a {\em Steiner node multicut} of $G$ if for every $1 \le i \le k$, there is at least one pair of vertices $\set{s, t}$ in $T_i$ such that $s$ and $t$ are contained in distinct components of $G - S$.
If $|T_i| = 2$ for every $1 \le i \le k$, $S$ is an ordinary node multicut of $G$.
This notion was introduced by Klein et al.~\cite{Klein::2012} and the problem of finding a minimum Steiner node multicut was studied in the literature~\cite{Bringmann::2016,Klein::2012}.

Let $H = (U, \mathcal E)$ be a hypergraph. A {\em transversal} of $H$ is a subset $S \subseteq U$ such that for every hyperedge $e \in \mathcal E$, it holds that $e \cap S \neq \emptyset$.
The problem of enumerating inclusion-wise minimal transversals, also known as dualizing monotone boolean functions, is one of the most challenging problems in this field.
There are several equivalent formulations of this problem and efficient enumeration algorithms developed for special hypergraphs.
However, the current best enumeration algorithm for this problem is due to Fredman and Khachiyan \cite{Fredman::1996}, which runs in quasi-polynomial time in the size of outputs, and no output-polynomial time enumeration algorithm is known.
In this section, we show that the problem of enumerating minimal Steiner node multicuts is as hard as this problem.

Let $H = (U, \mathcal E)$ be a hypergraph.
We construct a graph $G$ and sets of terminals as follows.
We begin with a clique on $U$.
For each $e \in \mathcal E$, we add a pendant vertex $v_e$ adjacent to $v$ for each $v \in e$ and set $T_e = \{v_e : v \in e\}$.
Note that $G$ is a split graph, that is, its vertex set can be partitioned into a clique $U$ and an independent set $\{v_e: e \in \mathcal E, v \in e\}$.
\begin{lemma}\label{lem:transversal}
    $S \subseteq U$ is a transversal of $H$ if and only if it is a Steiner node multicut of $G$.
\end{lemma}

\begin{proof}
    Suppose $S$ is a minimal transversal of $H$. Then, for each $e \in \mathcal E$, at least one vertex $v$ of $e$ is selected in $S$.
    Then, $v_e$ is an isolated vertex in $G - S$, and hence $S$ is a Steiner multicut of $G$.

    Conversely, suppose $S$ is a Steiner multicut of $G$.
    For each $T_e$, at least one pair of vertices $u_e$ and $v_e$ in $T_e$ are separated in $G - S$.
    Since $N(T_e)$ forms a clique, at least one of $u$ and $v$ is selected in $S$. Therefore, we have $S \cap e \neq \emptyset$.
\end{proof}

This lemma implies that if one can design an output-polynomial time algorithm for enumerating minimal Steiner node multicuts in a split graph, it allows us to do so for enumerating minimal transversals of hypergraphs.
For the problem of enumerating minimal Steiner edge multicuts, we could neither develop an efficient algorithm nor prove some correspondence as in Lemma~\ref{lem:transversal}.
We leave this question for future work.

\bibliographystyle{plain}
\bibliography{main}
\end{document}

%% file: macro.tex
\newcommand{\set}[1]{\left\{#1\right\}}

\newcommand{\size}[1]{\left| #1 \right|}
\newcommand{\order}[1]{O\hspace{-0.06cm}\left(#1\right)}

\newcommand{\Par}[1]{{\tt par}\left(#1\right)}

\newtheorem{lemma}{\bf Lemma}
\newtheorem{theorem}[lemma]{\bf Theorem}
\newtheorem{corollary}[lemma]{\bf Corollary}
\newtheorem{proposition}[lemma]{\bf Proposition}
\newtheorem{claim*}[lemma]{\bf Claim}

\Crefname{algocf}{Algorithm}{Algorithms}

\makeatletter
\@ifpackageloaded{algorithm2e}{
\SetKwProg{Fn}{Function}{}{}
\SetKwProg{Procedure}{Procedure}{}{}
\SetKwProg{Subprocedure}{Subprocedure}{}{}
\SetKwComment{tcc}{//}{}
\SetKw{Continue}{continue}
\SetKwFunction{Output}{Output}%
\SetKwFunction{EnumMC}{EMC}
\SetKwInOut{AlgInput}{Input}
\SetKwInOut{AlgOutput}{Output}
\SetKwInOut{AlgPrecondition}{Pre-conditions}
\SetKwInOut{AlgInvariant}{Invariants}
}{}
\makeatother

\newcommand{\ms}[1]{{\ifthenelse{\equal{#1}{}}{{ ms}}{{ ms}\xspace\left(#1\right)}}} 
\newcommand{\li}[1]{{\ifthenelse{\equal{#1}{}}{\ell}{\ell\xspace\left(#1\right)}}} 
\newcommand{\posms}[1]{{\ifthenelse{\equal{#1}{}}{{ p}}{{ p}\xspace\left(#1\right)}}} 

\newcommand{\dist}[2]{{\tt dist}\xspace(#1, #2)}
\newcommand{\depth}[1]{{\tt depth}\xspace(#1)}
\newcommand{\comp}[1]{{\ifthenelse{\equal{#1}{}}{{\tt \mu}}{{\tt \mu}\xspace\left(#1\right)}}}
\newcommand{\mcc}[2]{{\ifthenelse{\equal{#1}{}}{{\tt mcc}}{{\tt mcc}\xspace\left(#1, #2\right)}}}
\newcommand{\minimize}[1]{{\ifthenelse{\equal{#1}{}}{{\Gamma}}{{\Gamma}\xspace\left(#1\right)}}}

%% file: main.bbl
\begin{thebibliography}{10}

\bibitem{Arora::1999}
S.~Arora, D.~Karger, and M.~Karpinski.
\newblock Polynomial time approximation schemes for dense instances of
  {NP}-hard problems.
\newblock {\em Journal of Computer and System Sciences}, 58(1):193 -- 210,
  1999.

\bibitem{Avis::1996}
D.~Avis and K.~Fukuda.
\newblock Reverse search for enumeration.
\newblock {\em Discrete Applied Mathematics}, 65(1):21 -- 46, 1996.

\bibitem{Bateni::2012}
M.~Bateni, M.~Hajiaghayi, P.~N. Klein, and C.~Mathieu.
\newblock A polynomial-time approximation scheme for planar multiway cut.
\newblock In {\em Proceedings of the Twenty-Third Annual ACM-SIAM Symposium on
  Discrete Algorithms}, SODA ’12, page 639^^e2^^80^^93655, USA, 2012. Society
  for Industrial and Applied Mathematics.

\bibitem{Bouchitte:2002}
Vincent Bouchitt\'{e} and Ioan Todinca.
\newblock Treewidth and minimum fill-in: Grouping the minimal separators.
\newblock {\em SIAM J. Comput.}, 31(1):212^^e2^^80^^93232, 2002.

\bibitem{Bringmann::2016}
Karl Bringmann, Danny Hermelin, Matthias Mnich, and Erik~Jan [van Leeuwen].
\newblock Parameterized complexity dichotomy for {S}teiner {M}ulticut.
\newblock {\em Journal of Computer and System Sciences}, 82(6):1020 -- 1043,
  2016.

\bibitem{Chen::2004}
D.~Z. Chen and X.~Wu.
\newblock Efficient algorithms for k-terminal cuts on planar graphs.
\newblock {\em Algorithmica}, 38(2):299--316, Feb 2004.

\bibitem{Cohen::2008}
Sara Cohen, Benny Kimelfeld, and Yehoshua Sagiv.
\newblock Generating all maximal induced subgraphs for hereditary and
  connected-hereditary graph properties.
\newblock {\em J. Comput. Syst. Sci.}, 74(7):1147--1159, 2008.

\bibitem{Conte::2019}
Alessio Conte and Takeaki Uno.
\newblock New polynomial delay bounds for maximal subgraph enumeration by
  proximity search.
\newblock In {\em Proceedings of the 51st Annual {ACM} {SIGACT} Symposium on
  Theory of Computing, {STOC} 2019, Phoenix, AZ, USA, June 23-26, 2019}, pages
  1179--1190, 2019.

\bibitem{Calinescu::2000}
G.~C\v{a}linescu, H.~Karloff, and Y.~Rabani.
\newblock An improved approximation algorithm for multiway cut.
\newblock {\em Journal of Computer and System Sciences}, 60(3):564 -- 574,
  2000.

\bibitem{Cygan::2013}
M.~Cygan, M.~Pilipczuk, M.~Pilipczuk, and J.~O. Wojtaszczyk.
\newblock On multiway cut parameterized above lower bounds.
\newblock {\em ACM Trans. Comput. Theory}, 5(1), 2013.

\bibitem{Dahlhaus::1994}
E.~Dahlhaus, D.~S. Johnson, C.~H. Papadimitriou, P.~D. Seymour, and
  M.~Yannakakis.
\newblock The complexity of multiterminal cuts.
\newblock {\em SIAM J. Comput.}, 23(4):864^^e2^^80^^93894, 1994.

\bibitem{Diestel::2012}
R.~Diestel.
\newblock {\em Graph Theory, 4th Edition}, volume 173 of {\em Graduate texts in
  mathematics}.
\newblock Springer, 2012.

\bibitem{Feng::2014}
Jixing Feng, Xin Li, Eduardo~L. Pasiliao, and John~M. Shea.
\newblock Jammer placement to partition wireless network.
\newblock In {\em 2014 {IEEE} {GLOBECOM} Workshops, Austin, TX, USA, December
  8-12, 2014}, pages 1487--1492, 2014.

\bibitem{Fireman::2007}
L.~Fireman, E.~Petrank, and A.~Zaks.
\newblock New algorithms for simd alignment.
\newblock In {\em Compiler Construction}, pages 1--15, Berlin, Heidelberg,
  2007. Springer Berlin Heidelberg.

\bibitem{Fomin:2015}
Fedor~V. Fomin, Ioan Todinca, and Yngve Villanger.
\newblock Large induced subgraphs via triangulations and {CMSO}.
\newblock {\em SIAM J. Comput.}, 44(1):54--87, 2015.

\bibitem{Fredman::1996}
Michael~L. Fredman and Leonid Khachiyan.
\newblock On the complexity of dualization of monotone disjunctive normal
  forms.
\newblock {\em J. Algorithms}, 21(3):618--628, 1996.

\bibitem{Garg::1996}
Naveen Garg, Vijay~V. Vazirani, and Mihalis Yannakakis.
\newblock Approximate max-flow min-(multi)cut theorems and their applications.
\newblock {\em {SIAM} J. Comput.}, 25(2):235--251, 1996.

\bibitem{Garg::1997}
Naveen Garg, Vijay~V. Vazirani, and Mihalis Yannakakis.
\newblock Primal-dual approximation algorithms for integral flow and multicut
  in trees.
\newblock {\em Algorithmica}, 18(1):3--20, 1997.

\bibitem{Golumbic:2004}
Martin~Charles Golumbic.
\newblock {\em Algorithmic Graph Theory and Perfect Graphs (Annals of Discrete
  Mathematics, Vol 57)}.
\newblock North-Holland Publishing Co., NLD, 2004.

\bibitem{Guillemot::2011}
S.~Guillemot.
\newblock Fpt algorithms for path-transversal and cycle-transversal problems.
\newblock {\em Discrete Optimization}, 8(1):61--71, 2011.

\bibitem{Guo::2008}
Jiong Guo, Falk H{\"{u}}ffner, Erhan Kenar, Rolf Niedermeier, and Johannes
  Uhlmann.
\newblock Complexity and exact algorithms for vertex multicut in interval and
  bounded treewidth graphs.
\newblock {\em Eur. J. Oper. Res.}, 186(2):542--553, 2008.

\bibitem{Johnson::1988}
David~S. Johnson, Christos~H. Papadimitriou, and Mihalis Yannakakis.
\newblock On generating all maximal independent sets.
\newblock {\em Inf. Process. Lett.}, 27(3):119--123, 1988.

\bibitem{Kappes::2011}
J.~H. Kappes, M.~Speth, B.~Andres, G.~Reinelt, and C.~Schn\"{o}rr.
\newblock Globally optimal image partitioning by multicuts.
\newblock In {\em Proceedings of the 8th International Conference on Energy
  Minimization Methods in Computer Vision and Pattern Recognition},
  EMMCVPR’11, page 31^^e2^^80^^9344, Berlin, Heidelberg, 2011.
  Springer-Verlag.

\bibitem{Karger::2004}
D.~R. Karger, P.~Klein, C.~Stein, M.~Thorup, and N.~E. Young.
\newblock Rounding algorithms for a geometric embedding of minimum multiway
  cut.
\newblock {\em Math. Oper. Res.}, 29(3):436^^e2^^80^^93461, 2004.

\bibitem{Khachiyan::2008}
L.~Khachiyan, E.~Boros, K.~Borys, K.~Elbassioni, V.~Gurvich, and K.~Makino.
\newblock Generating cut conjunctions in graphs and related problems.
\newblock {\em Algorithmica}, 51(3):239^^e2^^80^^93263, 2008.

\bibitem{Khachiyan::2006}
Leonid Khachiyan, Endre Boros, Konrad Borys, Khaled~M. Elbassioni, Vladimir
  Gurvich, and Kazuhisa Makino.
\newblock Enumerating spanning and connected subsets in graphs and matroids.
\newblock In {\em Algorithms - {ESA} 2006, 14th Annual European Symposium,
  Zurich, Switzerland, September 11-13, 2006, Proceedings}, pages 444--455,
  2006.

\bibitem{Klein::2012}
P.~N. Klein and D.~Marx.
\newblock Solving planar k-terminal cut in ${O}(n^{c\sqrt{k}})$ time.
\newblock In {\em Proceedings of the 39th International Colloquium Conference
  on Automata, Languages, and Programming - Volume Part I}, ICALP'12, page
  569^^e2^^80^^93580, Berlin, Heidelberg, 2012. Springer-Verlag.

\bibitem{Kloks::1998}
T.~Kloks and D.~Kratsch.
\newblock Listing all minimal separators of a graph.
\newblock {\em SIAM J. Comput.}, 27(3):605^^e2^^80^^93613, June 1998.

\bibitem{Marx::2006}
D.~Marx.
\newblock Parameterized graph separation problems.
\newblock {\em Theor. Comput. Sci.}, 351(3):394^^e2^^80^^93406, 2006.

\bibitem{Marx::2012}
D.~Marx.
\newblock A tight lower bound for planar multiway cut with fixed number of
  terminals.
\newblock In {\em Proceedings of the 39th International Colloquium Conference
  on Automata, Languages, and Programming - Volume Part I}, ICALP’12, page
  677^^e2^^80^^93688, Berlin, Heidelberg, 2012. Springer-Verlag.

\bibitem{Marx::2014}
D{\'{a}}niel Marx and Igor Razgon.
\newblock Fixed-parameter tractability of multicut parameterized by the size of
  the cutset.
\newblock {\em {SIAM} J. Comput.}, 43(2):355--388, 2014.

\bibitem{Provan::1994}
J.~S. Provan and D.~R. Shier.
\newblock A paradigm for listing $(s, t)$-cuts in graphs.
\newblock {\em Algorithmica}, 15(4):351^^e2^^80^^93372, 1996.

\bibitem{Schwikowski::2002}
Benno Schwikowski and Ewald Speckenmeyer.
\newblock On enumerating all minimal solutions of feedback problems.
\newblock {\em Discret. Appl. Math.}, 117(1-3):253--265, 2002.

\bibitem{Shen::1997}
H.~Shen and W.~Liang.
\newblock Efficient enumeration of all minimal separators in a graph.
\newblock {\em Theoretical Computer Science}, 180(1):169 -- 180, 1997.

\bibitem{Stone::1977}
H.~S. Stone.
\newblock Multiprocessor scheduling with the aid of network flow algorithms.
\newblock {\em IEEE Trans. Softw. Eng.}, 3(1):85^^e2^^80^^9393, 1977.

\bibitem{Takata::2010}
K.~Takata.
\newblock Space-optimal, backtracking algorithms to list the minimal vertex
  separators of a graph.
\newblock {\em Discrete Applied Mathematics}, 158(15):1660 -- 1667, 2010.

\bibitem{Tsukiyama::1980}
S.~Tsukiyama, I.~Shirakawa, H.~Ozaki, and H.~Ariyoshi.
\newblock An algorithm to enumerate all cutsets of a graph in linear time per
  cutset.
\newblock {\em J. ACM}, 27(4):619^^e2^^80^^93632, 1980.

\bibitem{Uno::2003}
T.~Uno.
\newblock Two general methods to reduce delay and change of enumeration
  algorithms.
\newblock Technical report, National Institute of Informatics Technical Report
  E, 2003.

\bibitem{Xiao::2010}
M.~Xiao.
\newblock Simple and improved parameterized algorithms for multiterminal cuts.
\newblock {\em Theor. Comp. Sys.}, 46(4):723^^e2^^80^^93736, 2010.

\end{thebibliography}
